\documentclass[aps,pra,a4,superscriptaddress,preprintnumbers,twoside,twocolumn,floatfix]{revtex4}
\usepackage{stmaryrd}
\usepackage{graphicx}
\usepackage{color}
\usepackage{bm,bbold}
\usepackage{bbm}
\usepackage{enumerate}
\usepackage{color}
\graphicspath{figures} 
\usepackage{amsfonts, amsmath, amsthm, amssymb} 
\usepackage{mathtools}
\usepackage{dsfont}
\usepackage{array}
\usepackage[colorlinks,bookmarks=false,citecolor=blue,linkcolor=red,urlcolor=blue]{hyperref}
\usepackage{cleveref}
\usepackage{changes}

\newtheorem{theorem}{Theorem}
\newtheorem{lemma}{Lemma}

\begin{document}

\title{Coarse-grained self-testing}

\author{Ir\'en\'ee Fr\'erot}
\email{irenee.frerot@neel.cnrs.fr}
\affiliation{ICFO - Institut de Ciencies Fotoniques, The Barcelona Institute of Science and Technology, 08860 Castelldefels (Barcelona), Spain}
\affiliation{Max-Planck-Institut f\"ur Quantenoptik, Hans-Kopfermann-Stra{\ss}e 1, 85748 Garching, Germany}

\author{Antonio Ac\' in}
\affiliation{ICFO - Institut de Ciencies Fotoniques, The Barcelona Institute of Science and Technology, 08860 Castelldefels (Barcelona), Spain}
\affiliation{ICREA - Institucio Catalana de Recerca i Estudis Avan\c cats, Pg. Lluis Companys 23, 08010 Barcelona, Spain}

\begin{abstract}
Self-testing is a device-independent method that usually amounts to show that the maximal quantum violation of a Bell's inequality certifies a unique quantum state, up to some symmetries inherent to the device-independent framework. In this work, we enlarge this approach and show how a coarse-grained version of self-testing is possible in which physically relevant properties of a many-body system are certified. To this aim we study a Bell scenario consisting of an arbitrary number of parties and show that the membership to a set of (entangled) quantum states whose size grows exponentially with the number of parties can be self-tested. Specifically, we prove that a many-body generalization of the chained Bell inequality is maximally violated if and only if the underlying quantum state is equal, up to local isometries, to a many-body singlet. The maximal violation of the inequality therefore certifies any statistical mixture of the exponentially-many orthogonal pure states spanning the singlet manifold. 
\end{abstract}

\maketitle

Bell's inequalities (BIs) \cite{bell1964} constrain the correlation patterns achievable by any local-hidden-variables model~\cite{EPR1935}. Their violation by quantum-entangled states establishes the radically non-local nature of quantum statistical predictions \cite{brunneretal2014}. Analogously, the correlation patterns achievable by measuring quantum-entangled states are themselves also constrained -- by \textit{quantum} Bell's inequalities (qBIs) \cite{cirelson1980,brunneretal2014}. 
 Since the seminal work of Bell, BIs and qBIs have emerged as central concepts in the quest for laying the information-theoretic foundations of quantum physics \cite{popescuR1994,navascuesetal2015}; and in parallel, they have proved very powerful for characterizing quantum devices and protocols from minimal assumptions \cite{brunneretal2014}. The far-reaching conclusions which can be drawn from the violation of BIs, both from a foundational and from a quantum-certification perspective, rely fundamentally on their \textit{device-independent} nature. Namely, in contrast to quantum tomography \cite{paris2004quantum,flammiaL2011}, Bell tests involve solely the statistics of measurement results, and require no assumption about the Hilbert space of the system, neither about the observables which are actually being measured. 
 
 The device-independent nature of Bell tests makes them especially suited to robustly certify the new generation of quantum computers and simulators, where the qubits are effective two-level systems, requiring very careful calibration procedures in tomography protocols. For instance, the violation of a BI certifies the preparation of an entangled state, regardless of the correct calibration of the measurements. Quite remarkably, the maximal quantum violation of a BI -- reaching the quantum bound allowed by qBIs -- may allow one to certify, not only the generation of entanglement, but also the preparation of the specific quantum state and even the measurements necessarily performed to yield the observed correlations~\cite{mayersY2003}, up to symmetries inherent to the device-independent scenario, such as, e.g., local unitaries or complex conjugation. In essence, this so-called \textit{self-testing} phenomenon \cite{mayersY2003,supicB2020} realizes the simultaneous tomography of both a quantum state \textit{and} measurements, without prior assumptions on the system devices. 
 
While the self-testing phenomenon has been known for a long time for the celebrated Clauser-Horne-Shimony-Holt BI \cite{clauser_proposed_1969} (whose maximal violation self-tests the spin singlet $(|\uparrow \downarrow\rangle - |\downarrow \uparrow\rangle)/\sqrt{2}$ \cite{tsirelson1993,popescuR1992}), for multi-partite devices composed of $N\gg 1$ qubits it raises formidable challenges \cite{supicB2020}. Indeed, as the Hilbert space dimension grows exponentially with $N$, one can generically expect self-testing statements, which typically involve the fidelity with some target state, to be highly sensitive to noise, possibly exponentially, as $N$ increases. Instead of certifying a given target many-body quantum state \cite{takeuchiM2018,mckague10,baccarietal2020a}, it is natural in a many-body context to aim at self-testing less specific \cite{goh18,baccarietal2020,Makuta_2021}, yet physically relevant, global properties with a much better precision.

\begin{figure}
	\includegraphics[width=\linewidth]{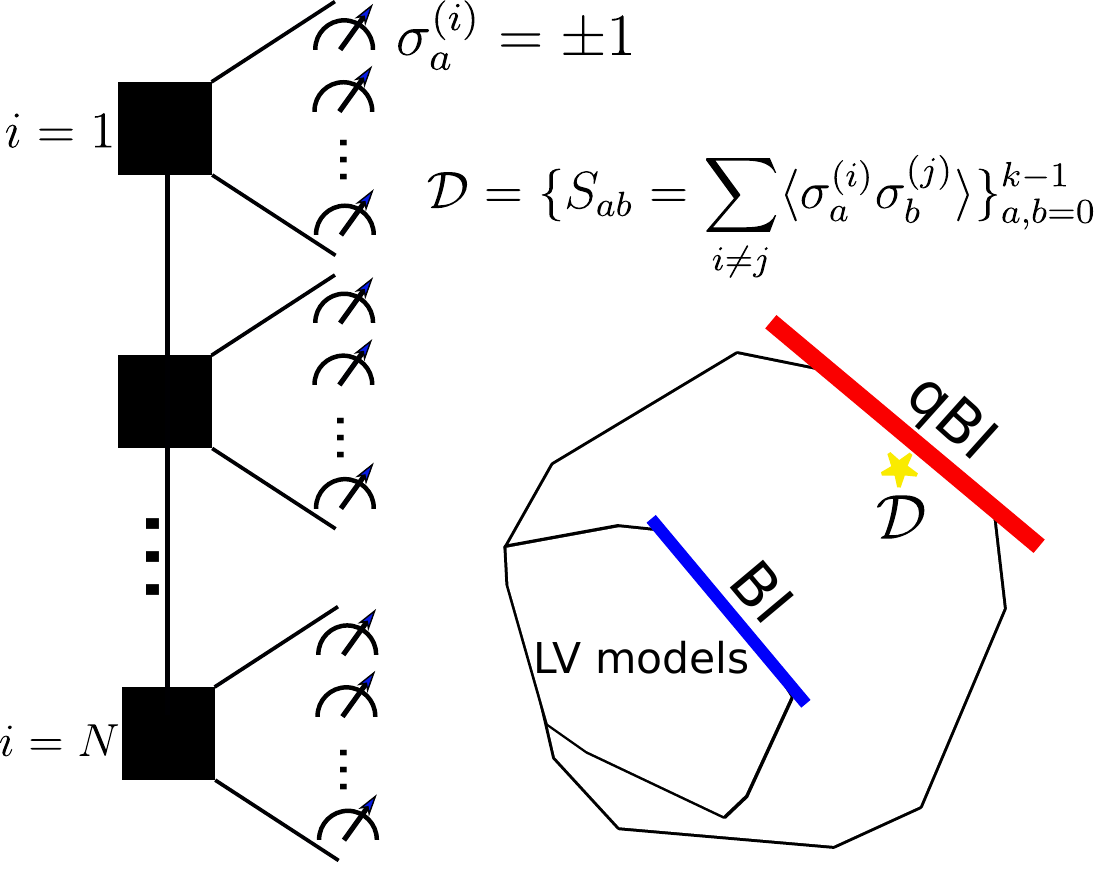}
	\caption{Bell scenario. We consider a scenario with $k \ge 3$ binary-outputs observables for each subsystem. In a Bell scenario, subsystems are treated as black boxes, generating some ouput $\sigma_a^{(i)} =\pm 1$ when the setting $a$ is implemented on subsystem $i$. Our self-testing procedure involves only very coarse-grain features of the correlations between these output: two-body correlations summed over all pairs of subsystems. If the corresponding data point (yellow star) reaches the quantum bound of our Bell's inequality (red solid line), a many-body singlet is self-tested.}
	\label{fig_setting}
\end{figure}

Here, we show that this form of coarse-grained certification is possible by self-testing the membership of the unknown quantum state $\hat\rho$ to the exponentially-degenerate subspace spanned by all many-body singlets, that is, zero-eigenstates of the total spin {$\hat {\bf J}^2 = [\sum_{i=1}^N \hat X^{(i)}/2]^2 + [\sum_{i=1}^N \hat Y^{(i)}/2]^2 + [\sum_{i=1}^N \hat Z^{(i)}/2]^2$ (with $\hat X^{(i)}$, $\hat Y^{(i)}$ and $\hat Z^{(i)}$ Pauli matrices for party $i$). Specifically, we construct a BI whose maximal quantum violation is obtained only by states $\hat\rho$ such that ${\rm Tr}[\hat \rho \hat {\bf J}^2]=0$}. 
We then show that the obtained certification is robust. From a certification perspective, many-body singlets are ground states of Heisenberg antiferromagnets \cite{Auerbach}, and emerge e.g. as low-energy states in simulators of Hubbard models \cite{koepselletal2019,chiuetal2019,Sun_2021}; they can also be prepared in atomic vapors using feedback schemes \cite{MorganExp}. From a fundamental perspective, our findings show that the boundaries of the set of quantum correlations, defined by qBIs, may correspond to exponentially degenerate entangled subspaces, and be reached by entangled many-body states of extensive entropy (here, any mixture of singlets, which span a subspace of dimension ${N \choose N/2} - {N \choose N/2-1} \sim 2^N \sqrt{8/(\pi N^3)}$ \cite{arecchietal1972}).

\noindent\textit{Definition of the Bell scenario.}
{We consider a situation where $k \ge 3$ possible measurements $\bar \sigma_a^{(i)}$ ($a\in\{0,1,\ldots k-1\}$) can be performed on each subsystem (or party) $i \in \{1, 2 \ldots N\}$ ($N$ is even), with $\pm 1$ as possible outcomes (Fig.~\ref{fig_setting}). In a device-independent scenario, the $\bar \sigma_a^{(i)}$ operators are not known -- in fact, even the Hilbert space over which they act is unknown. They obey the constraint of having $\pm 1$ as only eigenvalues, namely they square to the identity operator: $[\bar \sigma_a^{(i)}]^2= \bar{\mathbb{1}}$; furthermore, operators acting on different parties commute: $[\bar\sigma_a^{(i)}, \bar\sigma_b^{(j)}] = 0$ for all $a,b$ if $i \neq j$. Throughout this paper, the `bar' notation $\bar \sigma_a^{(i)}$ indicates that the operator acts on the uncharacterized subsystem $i$. Instead, the notation $\hat \sigma_a^{(i)}$ will denote a qubit operator acting on a two-level system attached to party $i$. In a practical implementation of the self-testing procedure, the parties are (effective) two-level systems, and the $\bar \sigma_a^{(i)}$ operators must correspond to projective measurements at equal angles in a given plane: $\bar \sigma_a^{(i)} \equiv \hat Z^{(i)} \cos(a\pi/k) + \hat X^{(i)} \sin(a \pi / k)$; this property is however not assumed, but instead is self-tested by the maximal violation of the BI. Finally, when discussing Bell's local-variable models, we denote as $\sigma_a^{(i)} = \pm 1$ the corresponding measurement outcome, which can be interpreted as a classical (Ising) spin \cite{Fine1982,BraunsteinC1990,frerotR2020}. 

In order to collect the data required for our self-testing procedure, the many-body system is identically prepared many times, varying the local measurement settings of the parties. Two-body correlators $\langle \bar\sigma_a^{(i)} \bar\sigma_b^{(j)} \rangle := {\rm Tr}[\bar \rho \bar \sigma_a^{(i)} \bar \sigma_b^{(j)} ]$ are then collected (for $i \neq j$).}  The following linear combination of two-body correlators (invariant under all permutations of the parties) is then formed~\cite{frerotR2020,mulleretal2020}: 
\begin{equation}
	{\cal B} =  \frac{2}{k}\sum_{a,b=0}^{k-1} S_{ab} \cos[\pi(a-b)/k] 
	\label{eq_BI}
\end{equation}
where $S_{ab} = \sum_{i \neq j} \langle \bar\sigma_a^{(i)} \bar\sigma_b^{(j)} \rangle$. For $k=3$, the quantity ${\cal B} $ is the sum over all pairs of the quantity involved in the so-called chained BI of Pearle and Braustein-Caves \cite{Pearle1970,BraunsteinC1990,frerotR2020}. The chained inequality (for any $k\ge 3$) is already known to self-test a spin singlet for $N=2$ \cite{supicetal2016}.

Introducing the collective operators $\bar{S_a} = \sum_{i=1}^N \bar \sigma_a^{(i)}$, we have $S_{ab} = \langle \bar{S_a} \bar{S_b} \rangle - \sum_i \langle \bar\sigma_a^{(i)} \bar\sigma_b^{(i)}\rangle$. We introduce the vector notations $[\bar{\bm \sigma}^{(i)}]^T=(\bar\sigma_0, \cdots \bar\sigma_{k-1})^{(i)}$ and $\bar{\bf S}^T = (\bar{S}_0, \ldots \bar{S}_{k-1})$, and the matrix $M_{ab} = (2/k)\cos[\pi(a-b)/k]$. We may then rewrite Eq.~\eqref{eq_BI} as:
\begin{equation}
	{\cal B}= \langle \bar{\bf S}^T M \bar{\bf S} \rangle - \sum_{i=1}^N \langle \bar{\bm \sigma}^T M \bar{\bm \sigma}\rangle^{(i)} ~.\label{eq_BI_M_expr}
\end{equation}
The matrix $M$ is diagonalized as $M={\bf c} {\bf c}^T + {\bf s} {\bf s}^T$, where ${\bf c}^T = \sqrt{2/k} [\cos(a\pi/k)]_{a=0}^{k-1}$ and ${\bf s}^T = \sqrt{2/k} [\sin(a\pi/k)]_{a=0}^{k-1}$ are normalized orthogonal vectors (${\bf s}^T{\bf c}=0$ and ${\bf s}^T{\bf s}={\bf c}^T{\bf c}=1$). As the matrix $M$ is a rank-2 projector, it is therefore semi-definite positive, so that $\langle \bar{\bf S}^T M \bar{\bf S} \rangle \ge 0$. 

{
\begin{lemma}(Classical bound of the Bell's inequality)
For any Bell's local-variable model, ${\cal B} \ge B_{\rm c}$, with the classical bound $B_{\rm c}=-\frac{2N}{k\sin^2[\pi/(2k)]}$.
\label{thm_classical_bound}
\end{lemma}
}
\begin{proof} {Recall that in this context, the $\sigma_a^{(i)}=\pm 1$ are just classical Ising spins. First, we notice that for $N$ even, one may always achieve $S_a=\sum_{i=1}^N \sigma_a^{(i)}=0$ for all $a \in \{0, \dots, k-1\}$, and therefore $\langle {\bf S}^T M {\bf S} \rangle = 0$. This is achieved by choosing a given configuration $\sigma_a$ on half of the parties, and the configuration $-\sigma_a$ on the remaining parties.} The classical bound is then found as: $B_{\rm c} = -N\max_{\bm \sigma \in \{\pm 1\}^k} ({\bm \sigma}^T M {\bm \sigma})= -(2N/k)\max_{{\bm \sigma} \in \{\pm 1\}^k} \left\vert \sum_{a=0}^{k-1} \sigma_a e^{ia\pi/k} \right\vert^2= -2N/(k\sin^2[\pi/(2k)])$, where the maximum is achieved by choosing, for instance, all $\sigma_a=+1$.
\end{proof}

{After deriving the classical bound for our BI, we move to the derivation of its maximal quantum violation.
\begin{lemma}(Quantum bound of the Bell's inequality)
  For any choice of measurements on any quantum state, ${\cal B} \ge B_{\rm q}$, with the quantum bound $B_{\rm q} = -Nk$. 
\end{lemma}
\begin{proof}
To prove that ${\cal B}\ge -Nk$, we construct the Bell operator $\bar{\cal B}$ (such that ${\cal B}=\langle \bar {\cal B} \rangle$), shifted by its claimed quantum bound $-Nk$. Using the property $Nk\bar{\mathbb{1}} = \sum_{i=1}^N \sum_{a=0}^{k-1} [\bar\sigma_a^{(i)}]^2 = \sum_{i=1}^N [\bar{\bm \sigma}^T \bar{\bm \sigma}]^{(i)}$, we find:
\begin{equation}
	\bar{\cal B} + Nk\bar{\mathbb{1}} = ({\bf c}^T \bar{\bf S})^2 + ({\bf s}^T \bar{\bf S})^2 
	+ \sum_{i=1}^N [\bar{\bm \sigma}^T (\mathbb{1} - M) \bar{\bm \sigma}]^{(i)}  ~. \label{eq_SOS_0}
\end{equation}
As mentioned, $\bar\sigma_a^{(i)}$ in this expression are generic operators in an arbitrary Hilbert space with the only constraint that $[\bar\sigma_a^{(i)}]^2=\bar{\mathbb{1}}$.
Since $M$ is a projector, so is $\mathbb{1}-M$, and hence $[\mathbb{1}-M]^2 = \mathbb{1}-M$. This decomposition therefore realizes a so-called sum-of-squares (SOS) decomposition of the Bell's inequality, showing that for any quantum state $|\psi\rangle$ and two-outcome local observables $\bar\sigma_a^{(i)}$, we have $\langle \psi |\bar{\cal B}|\psi\rangle +Nk\ge 0$. 
\end{proof}

After deriving the quantum bound, we want to show that it can be attained by measuring a many-body singlet; and conversely, that the maximal violation self-tests a many-body singlet. To prove this result, it is first convenient to identify conditions implied by the maximal quantum violation of the BI. To do so, we introduce the two operators:
\begin{subequations}
\label{eq_def_barZX}
\begin{align}
	\bar{Z}^{(i)} := \sqrt{2/k} ~ {\bf c}^T \bar{\bm \sigma}^{(i)} = \frac{2}{k}\sum_{a=0}^{k-1} \bar  \sigma_a^{(i)} \cos\left(\frac{a\pi}{k}\right) \label{eq_def_barZ}\\
	\bar{X}^{(i)} := \sqrt{2/k} ~ {\bf s}^T \bar{\bm \sigma}^{(i)} = \frac{2}{k}\sum_{a=0}^{k-1}\bar  \sigma_a^{(i)}  \sin\left(\frac{a\pi}{k}\right)\label{eq_def_barX} ~.
\end{align}
\end{subequations}
Defining the (hermitian) operators:
\begin{subequations}
\begin{align}
&\bar S_z := \sum_{i=1}^N \bar{Z}^{(i)}~~~;~~~\bar S_x := \sum_{i=1}^N \bar{X}^{(i)} \\
&\bar A_a^{(i)} := \bar \sigma_a^{(i)} - \frac{2}{k}\sum_{b=0}^{k-1} \bar \sigma_b^{(i)} \cos\left(\pi\frac{a-b}{k}\right) ~, \label{eq_def_Abar}
\end{align}
\end{subequations}
the SOS decomposition Eq.~\eqref{eq_SOS_0} reads:
\begin{equation}
	\bar{\cal B} + Nk\bar{\mathbb{1}} = \frac{k}{2}(\bar{S}_z^2 + \bar{S}_x^2) 
	+ \sum_{i=1}^N \sum_{a=0}^{k-1} [\bar{A}_a{(i)}]^2  ~.
	\label{eq_SOS}
\end{equation}The following Lemma \ref{lemma_paulis} shows that $\bar{Z}^{(i)}$ and $\bar{X}^{(i)}$ [Eq.~\eqref{eq_def_barZX}] act on a quantum state reaching the quantum bound like the Pauli matrices $\hat Z^{(i)}$ and $\hat X^{(i)}$ act on a qubit:
\begin{lemma}
Let $|\psi\rangle$ be a quantum state s.t. $\langle \psi |\bar{\cal B}|\psi\rangle = - Nk$. We have:\begin{subequations}
\label{eq_properties_XZ}
\begin{align}
[\bar{Z}^{(i)}]^2 |\psi \rangle = [\bar{X}^{(i)}]^2|\psi \rangle = |\psi\rangle ~,\\ 
(\bar{Z}\bar{X} + \bar{X}\bar{Z})^{(i)}|\psi \rangle = 0 ~.
\end{align}
\end{subequations}
\label{lemma_paulis}
\end{lemma}}
\begin{proof}
First, it is straighforward to verify that if the parties hold qubits with measurement operators $\bar\sigma_a^{(i)} = \hat Z^{(i)} \cos(a\pi/k) + \hat X^{(i)} \sin(a\pi/k)$, then indeed in Eq.~\eqref{eq_def_barZX} $\bar{Z}^{(i)} = \hat Z^{(i)}$ and $\bar{X}^{(i)} = \hat{X}^{(i)}$. In order to prove Lemma \ref{lemma_paulis}, we shall exploit conditions imposed by the SOS decomposition Eq.~\eqref{eq_SOS} being zero, {and that are necessary to reach the quantum bound}: $\sum_p \langle \psi | \bar O_p^\dagger \bar O_p |\psi \rangle = 0$ iff for all $p$, $\bar O_p |\psi \rangle = 0$. 
Our proof proceeds in two steps: 1) we show that $[\bar{Z}^2 + \bar{X}^2]^{(i)}|\psi \rangle = 2|\psi\rangle$; then 2) that $[\bar{Z}^2 - \bar{X}^2]^{(i)}|\psi \rangle = 0 = [\bar{Z} \bar{X} + \bar{X} \bar{Z}]^{(i)}|\psi \rangle$.

1) For a moment, we drop the superscript $(i)$; the following equalities hold for each subsystem. We have $\bar{Z}^2 + \bar{X}^2 = (2/k)[({\bf c}^T \bar{\bm \sigma})({\bf c}^T \bar{\bm \sigma}) + ({\bf s}^T \bar{\bm \sigma})({\bf s}^T \bar{\bm \sigma})] = (2/k)\bar{\bm \sigma}^T[{\bf c}{\bf c}^T + {\bf s}{\bf s}^T]\bar{\bm \sigma} = (2/k)\bar{\bm \sigma}^T M \bar{\bm \sigma}$. The SOS implies that $\bar{\bm \sigma}^TM\bar{\bm \sigma} |\psi \rangle = \bar{\bm \sigma}^T \bar{\bm \sigma} |\psi \rangle$, which is equal to $k|\psi \rangle$. This shows that $(\bar{Z}^2 + \bar{X}^2)|\psi \rangle = 2|\psi\rangle$.

2) The second part of the proof starts from $\bar A_a^{(j)}|\psi\rangle = 0$, where $\bar A_a^{(j)}$ is defined in Eq.~\eqref{eq_def_Abar}. Inserting the definition of $\bar{Z}^{(j)}$ and $\bar{X}^{(j)}$ [Eq.~\eqref{eq_def_barZX}], we obtain:
\begin{equation}
 	[\bar{Z} \cos(a\pi/k)   + \bar{X}\sin(a\pi/k)]^{(j)} |\psi \rangle = \bar \sigma_a^{(j)} |\psi \rangle ~.\label{eq_self_test_measurements}
\end{equation}
We use Eq.~\eqref{eq_self_test_measurements} in the form:
\begin{equation}
	e^{2i\pi a /k}(\bar{Z} - i\bar{X})^{(j)}|\psi \rangle = [2e^{i\pi a /k}\bar \sigma_a - (\bar{Z} + i\bar{X})]^{(j)} |\psi \rangle ~.
	\label{eq_starting_XZ_ZX}
\end{equation}
We define the operator $\bar R_a^{(j)} = [2e^{i\pi a /k} \bar\sigma_a - (\bar{Z} + i\bar{X})]^{(j)}$, apply $\bar R_a^{(j)}$ to the last equality, and sum over $a$. On the r.h.s, we simply have $\sum_{a=0}^{k-1} \bar R_a^{(j)} \bar R_a^{(j)} = -k[\bar{Z}^{(j)} + i\bar{X}^{(j)}]^2$. On the l.h.s, we need to evaluate $\bar R_a^{(j)}(\bar{Z} - i\bar{X})^{(j)}|\psi \rangle$. From the SOS, we have that $\sum_{j} \bar{Z}^{(j)} |\psi \rangle = \bar{S}_z |\psi \rangle = 0$ and also $\sum_{j} \bar{X}^{(j)} |\psi \rangle = \bar{S}_x |\psi \rangle = 0$. We define $\bar \Sigma_Z^{(j)} = \sum_{j'\neq j} \bar{Z}^{(j)}$ and $\bar \Sigma_X^{(j)} = \sum_{j'\neq j} \bar{X}^{(j)}$, so that $(\bar{Z} - i\bar{X})^{(j)}|\psi \rangle = -(\bar \Sigma_Z - i \bar \Sigma_X)^{(j)} |\psi \rangle$. As the operators $\bar \Sigma_Z^{(j)}$ and  $\bar \Sigma_X^{(j)}$ act on subsystems $j'\neq j$, they commute with any operator acting on $j$. Therefore:
\begin{eqnarray}
	\bar R_a^{(j)}(\bar{Z} - i\bar{X})^{(j)}|\psi \rangle = -\bar R_a^{(j)}(\bar \Sigma_Z - i\bar \Sigma_X)^{(j)}|\psi \rangle \\
	= -(\bar \Sigma_Z - i\bar \Sigma_X)^{(j)}\bar R_a^{(j)}|\psi \rangle \\
	= -(\bar \Sigma_Z - i \bar \Sigma_X)^{(j)}e^{2i\pi a /k}(\bar{Z} - i\bar{X})^{(j)}|\psi \rangle \\
	= -e^{2i\pi a /k}(\bar{Z} - i\bar{X})^{(j)}(\bar\Sigma_Z - i\bar\Sigma_X)^{(j)}|\psi \rangle \\
	= e^{2i\pi a /k}(\bar{Z} - i\bar{X})^{(j)}(\bar{Z} - i\bar{X})^{(j)}|\psi \rangle
\end{eqnarray}
We now have $\bar R_a^{(j)}e^{2i\pi a /k}(\bar{Z} - i\bar{X})^{(j)}|\psi \rangle = e^{4i\pi a /k}[\bar{Z}^{(j)} - i\bar{X}^{(j)}]^2 |\psi \rangle$. Summing over $a$, and using that $\sum_{a=0}^{k-1} e^{4i\pi a /k} = 0$ for all $k \ge 3$, we conclude that $[\bar{Z}^{(j)} + i\bar{X}^{(j)}]^2|\psi \rangle = 0$, namely: $[\bar{Z}^2 - \bar{X}^2 + i(\bar{Z} \bar{X} + \bar{X} \bar{Z})]^{(j)}|\psi \rangle = 0$.
Following the same reasoning interchanging the role of $\bar{Z}^{(j)} + i\bar{X}^{(j)}$ and $\bar{Z}^{(j)} - i\bar{X}^{(j)}$ in Eq.~\eqref{eq_starting_XZ_ZX}, we also have $[\bar{Z}^2 - \bar{X}^2- i(\bar{Z} \bar{X} + \bar{X} \bar{Z})]^{(j)}|\psi \rangle = 0$. Taking the sum and difference of these two inequalities, we obtain the announced result 2): for all subsystems $j$, $[\bar{Z}^2 - \bar{X}^2]^{(j)}|\psi \rangle = 0 = [\bar{Z} \bar{X} + \bar{X} \bar{Z}]^{(j)}|\psi \rangle$. This concludes our proof of Eq.~\eqref{eq_properties_XZ}, namely that $\bar{Z}^{(i)}$ and $\bar{X}^{(i)}$ act on an unkown quantum state state $|\psi\rangle$ maximally violating the Bell's inequality as the Pauli operators $\hat Z^{(i)}$ and $\hat X^{(i)}$ act on a qubit.
\end{proof}

We have now derived all the ingredients needed to prove the main result of this work, namely the self-testing of many-body singlet states.

  \begin{theorem}
The maximal violation of the BI~\eqref{eq_BI}, ${\cal B} = -Nk$, is attained iff the state is (up to a local isometry mapping each party onto a qubit) a many-body singlet, with measurements $\hat\sigma_a^{(i)} = \hat Z^{(i)} \cos (a\pi/k) + \hat X^{(i)} \sin (a\pi/k)$.
  	\label{thm_self_test}
  \end{theorem}
\begin{proof} ($\Leftarrow$): For qubit measurements $\hat\sigma_a^{(i)} = \hat Z^{(i)} \cos \theta_a + \hat X^{(i)} \sin \theta_a$, we have ${\cal B} = 4\sum_{ab} M_{ab} \langle \hat J_a \hat J_b \rangle -N \sum_{ab}M_{ab} \cos(\theta_a - \theta_b)$, where $\hat J_a = (1/2)\sum_i \hat \sigma_a^{(i)}$ is a collective spin component. For a many-body singlet, we have $\langle \hat J_a \hat J_b \rangle = 0$ for any measurement directions.  Choosing $\theta_a = a\pi /k$, we obtain $ {\cal B}  = -(Nk/2) \sum_{ab}M_{ab}^2 = -(Nk/2){\rm Tr}(M^2) = -Nk$.

{ ($\Rightarrow$): To prove the self-testing statement of Theorem \ref{thm_self_test}, we use Lemma \ref{lemma_paulis} to construct the the so-called partial SWAP gate, the said isometry mapping each party onto a qubit.} As illustrated in Fig.~\ref{fig_swap}, an extra qubit in state $|+\rangle=(|0\rangle + |1\rangle)/\sqrt{2}$ is attached to each party, and we apply locally:
\begin{eqnarray}
	\Phi_i[\bar{\rho}_i] = {\rm Tr}_{\rm b.b}[\hat U^{(i)} (|+\rangle \langle +| \otimes \bar{\rho}_i) (\hat U^{(i)})^{\dagger}] \\
{\rm with }~~\hat U^{(i)} = (c {\bar X}^{(i)})(\hat H \otimes \bar{\mathbb{1}})(c{\bar Z}^{(i)}) \nonumber
\end{eqnarray}

where $\hat H$ is the Hadamard gate on the ancillary qubit, and ${\rm Tr}_{\rm b.b}$ denotes a trace over black box $i$, namely the local degrees of freedom which are not the qubit; and operator $c\bar{X}^{(i)}$ [resp.~$c\bar{Z}^{(i)}$] is a control-$\bar{X}^{(i)}$ gate [resp.~control-$\bar{Z}^{(i)}$], where the control is on the state of qubit $i$. The SWAP gate consists in applying the partial SWAP gate to each party: $\Phi = \otimes_{i=1}^N \Phi_i$.

\begin{figure}
	\includegraphics[width=\linewidth]{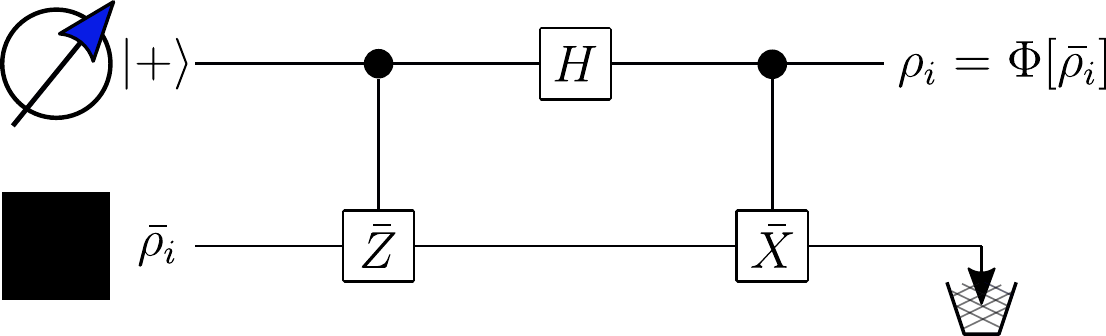}
	\caption{Partial SWAP gate. To each party (black box), initially in the (unknown) state $\bar{\rho_i}$, an ancillary qubit is attached and initialized in the state $|+\rangle=(|0\rangle + |1\rangle)/\sqrt{2}$. The partial SWAP gate is realized by successively applying a control-$\bar{Z}^{(i)}$ gate, a Hadamard gate on the qubit, and a control-$\bar{X}^{(i)}$ gate [$\bar{Z}^{(i)}$ and $\bar{X}^{(i)}$ are defined in Eq.~\eqref{eq_def_barZX}]. The gate is completed by tracing over the black-box degrees of freedom.
	}
	\label{fig_swap}
\end{figure}

{ Now, let $\bar{\rho}$ be a quantum state reaching the quantum bound of the BI. As discussed in \cite{supp_mat}, this implies that $\bar{Z}^{(i)}$ and $\bar{X}^{(i)}$ act as unitaries, so that $\Phi_i$ (and $\Phi$) are isometries; in the general case, $\bar{Z}^{(i)}$ and $\bar{X}^{(i)}$ must be properly regularized \cite{supp_mat}. We then define the $N$-qubit state $\hat \rho:= \Phi[\bar{\rho}]$, and show that $\hat \rho$ is a many-body singlet.} A key property is that applying $\bar{X}^{(i)}$ [resp.~$\bar{Z}^{(i)}$] to black box $i$ before applying the partial-SWAP gate is equivalent to applying the Pauli operators $\hat X^{(i)}$ [resp.~$\hat Z^{(i)}$] to qubit $i$ after the partial-SWAP gate: $\Phi[\bar{X}^{(i)}\bar{\rho}] = \hat X^{(i)} \hat \rho$ (and similarly for $Z$). This can be checked by direct inspection of the partial SWAP gate, using Lemma \ref{lemma_paulis}. Maximal violation of the BI also implies, via Eq.~\eqref{eq_SOS}, $\sum_{i=1}^{N} \bar{Z}^{(i)} \bar{\rho} = \bar{S}_z \bar{\rho} = 0 = \sum_{i=1}^{N} \bar{X}^{(i)}\bar{\rho} = \bar{S}_x \bar{\rho}$. The same property holds after the SWAP gate: $\hat J_x \hat \rho  = \hat J_z \hat \rho = 0$,
where $\hat J_x = \sum_{i=1}^N \hat X^{(i)}/2$ and $\hat J_z = \sum_{i=1}^N \hat Z^{(i)}/2$ are collective spin observables for the $N$ qubits. Since $[\hat J_z, \hat J_x]=i\hat J_y$, this implies in turn $\hat J_y \hat \rho=0$, or equivalently ${\rm Tr}[\hat {\bf J}^2 \hat \rho] = 0$. 
Hence, maximal quantum violation of the BI self-tests a many-body singlet, as announced as the main result of our paper. Additionally, the measurements are also self-tested: as a consequence of Eq.~\eqref{eq_self_test_measurements}, saturating the quantum bound implies that the $\bar \sigma_a^{(i)}$ operators act on the state like equispaced spin measurements in the $xz$ plane (a property already proved for the chained inequality \cite{Pearle1970,BraunsteinC1990} in the $N=2$ case \cite{supicetal2016}). This concludes the proof of Theorem \ref{thm_self_test}.
\end{proof}

\noindent\textit{Robustness.}
A self-testing statement can also be obtained for non-maximal quantum violations. A robust self-testing statement consists in establishing that if $\langle \bar {\cal B} \rangle + Nk \le kN \epsilon$, then it is possible to define global spin operators such that $\langle \hat J_z^2 + \hat J_x^2 \rangle \le f(N, \epsilon)$, with a certain function $f$ such that $f(N, 0) = 0$. In \cite{supp_mat}, we establish such a robustness statement, with a bound of the form:
\begin{equation}
\label{robbound}
	\langle \hat J_z^2 + \hat J_x^2 \rangle \le \frac{N^2 \epsilon}{4} \left[\sqrt{2/N} + \sqrt{r}\right]^2
\end{equation}
with $r \approx 752$ a numerical constant. For $N \to \infty$, the bound on $\langle \hat J_z^2 + \hat J_x^2 \rangle$ scales as $O(N^2\epsilon)$

It is an open question whether this scaling could be improved. When considering the qubit realisation leading to the maximal quantum violation, from Eq.~\eqref{eq_SOS}, it is clear that if one measures a $N$-qubit state along $\hat \sigma_a^{(i)} = \hat Z^{(i)} \cos(a\pi/k) + \hat{X}^{(i)} \sin(a\pi/k)$, then one obtains:
\begin{equation}
	{\cal B} + Nk = \frac{k}{2}\left\langle \left[\sum_{i=1}^N \hat Z^{(i)}\right]^2 + \left[\sum_{i=1}^N \hat X^{(i)}\right]^2 \right\rangle = 2k \langle \hat J_z^2 + \hat J_x^2 \rangle ~.
\end{equation}
In a real experiment, the preparation of a many-body singlet is not perfect, and we have $2\langle \hat J_z^2 + \hat J_x^2 \rangle = \epsilon N$ with a small, but finite $\epsilon$, yielding ${\cal B} + Nk = kN \epsilon$, that is a $O(N\epsilon)$ scaling. Although this scaling assumes noiseless qubit measurements, it suggests that the robustness bound~\eqref{robbound} might be improved, either using analytical techniques such as in \cite{supp_mat}, or resorting to numerical techniques~\cite{yangetal2014}.

\noindent\textit{Generalization.} The self-testing scheme we have described can be extended to certify more general coarse-grained features of the quantum state. Introducing arbitrary local phases $\phi_i$, one considers the following quantity:
\begin{equation}
	{\cal B} =  \frac{2}{k}\sum_{a,b=0}^{k-1} \sum_{i \neq j} \langle \bar\sigma_a^{(i)} \bar\sigma_b^{(j)} \rangle \cos[\pi(a-b)/k + \phi_i - \phi_j]  ~.
	\label{eq_BI_phases}
\end{equation}
In \cite{supp_mat}, we show that the above derivation presented for $\phi_i=0$ can be extended to arbitrary $\phi_i$'s. The quantum bound for Eq.~\eqref{eq_BI_phases} is also $-Nk$ and reaching it self-tests the property: $\sum_{i,j}\cos(\phi_i - \phi_j)\langle \hat X^{(i)}\hat X^{(j)} + \hat Z^{(i)} \hat Z^{(j)}\rangle=0$. The self-tested measurements are the same as for $\phi_i=0$ (namely: equally-spaced qubit measurements in the $xz$ plane). Equivalently, the BIs described by Eq.~\eqref{eq_BI_phases} allow self-testing a many-body singlet on which local rotations by angles $\phi_i$ around the $y$ axis have been performed, in the same measurement setting.

\noindent\textit{Conclusions.}
The certification of quantum many-body systems is a timely problem for which the device-independent framework provides useful tools, as the certification does not rely on any assumptions on the devices. Self-testing, possibly the strongest form of device-independent certification, has so far focused on certifying unique, modulo some local operations, quantum state and measurement operators. However, in a many-body context, it may be useful to consider coarse-grained variants of self-testing, in which only some physical properties of interest are certified from the observed statistics. In this work we demonstrate the validity of this approach by showing how to self-test, in a robust way, many-body singlets, defined by the sole condition ${\rm Tr}[\hat \rho (\hat J_x^2 + \hat J_y^2 + \hat J_z^2)] = 0$ with $\hat {\bf J}$ the collective spin. This very coarse-grain condition can be met by states which are statistical mixtures of exponentially-many orthogonal pure states, spanning the singlet manifold. To our knowledge, this represents the first self-testing example of this kind, {opening a new paradigm for the certification of many-body quantum devices. 

In itself, that such coarse-grain certification is simply possible represents an intriguing result. Our specific self-testing scheme could already be implemented in several platform able to implement individual qubit measurements, and to prepare many-body singlets (either by cooling a trapped Fermi gas interacting via a Hubbard-like Hamiltonian \cite{koepselletal2019,chiuetal2019}, or by adiabatically preparing the ground state of an antiferromagnetic Heisenberg model in a qubit ensemble \cite{Sun_2021}), which would demonstrate a supreme form of control over an entangled quantum many-body system. More generally, it is an outstanding question to understand the possibilities and limitations of this coarse-grain approach, and identify other relevant properties, of the form $\langle \hat O\rangle=0$ with $\hat O$ a positive semidefinite many-body observable, that can be robustly self-tested. In our example, a key aspect is that the matrix $M$, defining the BI via Eq.~\eqref{eq_BI_M_expr}, is a projector. This suggests to replace $M$ by a more general projector, or by a more general positive semidefinite matrix. A second extension is to introduce spatial modulation in the coefficients of the BI, and we presented a first result in this direction in Eq.~\eqref{eq_BI_phases}. 
}

\acknowledgments{We acknowledge support from the Government of Spain (FIS2020-TRANQI and Severo Ochoa CEX2019-000910-S), Generalitat de Catalunya (CERCA, AGAUR SGR 1381 and QuantumCAT), the Fundaci{\'o} Cellex and Fundaci{\'o} Mir-Puig through an ICFO-MPQ Postdoctoral Fellowship, ERC AdG CERQUTE, the AXA Chair in Quantum Information Science.}
\bibliography{biblio_1}

\appendix
\section{Robustness}
\label{app_robustness}

In this Appendix, we derive the robustness bound for the self-testing of the many-body singlet state.

\subsection{SWAP unitary}
The SWAP unitary is based on the operators $\bar{Z}^{(i)} = (2/k) \sum_{a=0}^{k-1}  \sigma_a^{(i)} \cos(a\pi/k)$ and $\bar{X}^{(i)} = (2/k) \sum_{a=0}^{k-1}   \sigma_a^{(i)}\sin(a\pi/k)$, where $i=1, \dots N$ labels the $i$-th party. To make the SWAP gate unitary, these operators must first be regularized in order to make them unitary. In the main text, we did not need to regularize these operators, as their unitary character was an implicit consequence of the maximal violation of the Bell's inequality. The regularization is obtained in the following manner. First, we notice that $\bar{Z}^{(i)}$ and $\bar{X}^{(i)}$ are hermitian by construction. They therefore admit the following eigenvalue decomposition: $\bar{Z}^{(i)} = \sum_\alpha P_\alpha \lambda_\alpha$ (and similarly for $\bar{X}^{(i)}$), where $\{P_\alpha\}$ are paiwise orthogonal projectors ($P_\alpha P_\beta = P_\alpha \delta_{\alpha, \beta}$) summing to the identity ($\sum_\alpha P_\alpha = \mathbb{1}$), and $\lambda_\alpha \in \mathbb{R}$ are the corresponding eigenvalues. We then define $\tilde{Z}^{(i)} = \sum_\alpha P_\alpha {\rm sign}(\lambda_\alpha)$ (with the convention ${\rm sign}(0) = 1$), which are unitary operators ($[\tilde{Z}^{(i)}]^2 = \mathbb{1}$, and similarly $[\tilde{X}^{(i)}]^2 = \mathbb{1}$). \\
Adding an ancillary qubit initialized in $|+\rangle_i$ to each party, the partial SWAP gate is defined as $\Phi_i = (c\tilde{X}^{(i)})(H \otimes \mathbb{1})(c\tilde{Z}^{(i)})$, where $H$ is the Hadamard gate on the qubit, and $cA^{(i)}$ is the control-$A^{(i)}$ gate (with $A=X,Z$), where the control is on the qubit state. $\Phi_i$ is a unitary transformation. The full SWAP gate is obtained by applying the partial SWAP gate to all parties: $U_\Phi = \otimes_{i=1}^N \Phi_i$, which is also unitary. Notice that the definition of the SWAP gate used in the robustness proof is slightly different from the definition used in the main text: in the main text, we introduced a partial trace over the black-box degrees of freedom. Here, in contrast, we do not take the partial trace, and consider instead a unitary transformation. The two definitions lead however to equivalent statements. 

\subsection{Definition of the robustness problem}
We define the collective spin observables for the qubits: 
\begin{eqnarray}
	S_x := \sum_{i=1}^N X^{(i)} \otimes \mathbb{1} \\	
	S_z := \sum_{i=1}^N Z^{(i)} \otimes \mathbb{1}
\end{eqnarray}
where $X^{(i)}$ and $Z^{(i)}$ are the local Pauli matrices, and the identity operator acts on the black-box degrees of freedom. Similarly, we define the operators: 
\begin{eqnarray}
	\bar{S}_x := \sum_{i=1}^N  \mathbb{1} \otimes \bar{X}^{(i)} \\	
	\bar{S}_z := \sum_{i=1}^N  \mathbb{1} \otimes \bar{Z}^{(i)}
\end{eqnarray}
where here the identity acts on the ancillary qubits. The shifted Bell operator is:
\begin{equation}
	{\cal B} + Nk = \frac{k}{2} (\bar{S}_x^2 + \bar{S}_z^2) 
	+ \sum_{i=1}^N ({\bm \sigma}^T [\mathbb{1} - M] {\bm \sigma} )^{(i)}  ~.
	\label{eq_SOS_appendix}
\end{equation}
If the violation of the Bell's inequality is such that $\langle {\cal B} + Nk \rangle \le Nk\epsilon$, then we have that $\langle \bar{S}_x^2 + \bar{S}_z^2 \rangle \le 2N\epsilon$. More explicitly, denoting as $|\Psi\rangle$ the (unknown) quantum state of the black-box degrees of freedom, we have:
\begin{equation}
	\langle \Psi | \langle +|^{\otimes N} (\bar{S}_x^2 + \bar{S}_z^2)|+\rangle^{\otimes N} |\Psi \rangle \le 2N\epsilon ~.
\end{equation}
The goal is in turn to establish a non-trivial bound on $\langle (S_x^2 + S_z^2) \rangle$ after application of the SWAP gate, namely to bound the quantity ${\cal Q}$:
\begin{equation}
	{\cal Q} := \langle \Psi | \langle +|^{\otimes N} U_\Phi^\dagger (S_x^2 + S_z^2)U_\Phi |+\rangle^{\otimes N} |\Psi \rangle \le ~?
\end{equation}
\textit{Strategy of the proof.} The strategy is to show that the operator $S_x^{(\Phi)} := U_\Phi^\dagger S_x U_\Phi$ is ``close''  to the operator $\bar{S}_x$ (and similarly for $S_z^{(\Phi)} := U_\Phi^\dagger S_z U_\Phi$ and $\bar{S}_z$), in some sense to define, and to use some form of the triangle inequality to obtain a non-trivial bound. Specifically, defining $\langle \cdot \rangle := \langle \Psi | \langle +|^{\otimes N} \cdot|+\rangle^{\otimes N} |\Psi \rangle$, we use that for any operators $A_1$ and $A_2$: $\langle A_1^\dagger A_1 \rangle = \langle A_2^\dagger A_2 \rangle + \langle (A_1-A_2)^\dagger (A_1-A_2) \rangle	+ \langle A_2^\dagger (A_1-A_2) \rangle + \langle (A_1-A_2)^\dagger A_1 \rangle$.
Using the Cauchy-Schwarz inequality, we have $|\langle A_2^\dagger (A_1-A_2) \rangle| \le \sqrt{\langle A_2^\dagger A_2 \rangle \langle (A_1 - A_2)^\dagger (A_1-A_2) \rangle}$, so that we have $\langle A_1^\dagger A_1 \rangle \le [\sqrt{\langle A_2^\dagger A_2 \rangle} + \sqrt{\langle (A_1-A_2)^\dagger (A_1-A_2) \rangle}]^2$. Applying this inequality with $(A_1,A_2)=(S_x^{(\Phi)}, \bar{S}_x)$ and $(A_1,A_2)=(S_z^{(\Phi)}, \bar{S}_z)$, we obtain:
\begin{eqnarray}
	{\cal Q} = \langle [S_x^{(\Phi)}]^2 + [S_z^{(\Phi)}]^2 \rangle \le \left[\sqrt{\langle \bar{S}_x^2 + \bar{S}_z^2 \rangle} \right. \nonumber \\
	+ \left. \sqrt{\langle [S_x^{(\Phi)}-\bar{S}_x]^2 + [S_z^{(\Phi)}-\bar{S}_z]^2 \rangle} \right]^2
	\label{eq_quantity_Q}
\end{eqnarray}
We already know that $\langle \bar{S}_x^2 + \bar{S}_z^2 \rangle \le 2N\epsilon$, so that it remains to bound $\langle [S_x^{(\Phi)}-\bar{S}_x]^2 + [S_z^{(\Phi)}-\bar{S}_z]^2 \rangle$. We will find a bound of the form $\langle [S_x^{(\Phi)}-\bar{S}_x]^2 + [S_z^{(\Phi)}-\bar{S}_z]^2 \rangle \le rN^2 \epsilon $ for some numerical factor $r \approx 752$, which implies a robustness statement of the form:
\begin{equation}
	{\cal Q} \le  \left[\sqrt{2N\epsilon} + \sqrt{rN^2 \epsilon}\right]^2 \approx rN^2\epsilon ~.
\end{equation}

\subsection{Bounding total spin fluctuations after the SWAP gate}
One has to evaluate the squared norm of: $[S_z^{(\Phi)} - \bar{S}_z] |+\rangle^{\otimes N} |\Psi\rangle$, and of $[S_x^{(\Phi)} - \bar{S}_x] |+\rangle^{\otimes N} |\Psi\rangle$.\\

\subsubsection{Two useful Lemmas}
To derive the robustness bound, we shall make extensive use of two inequalities:
\begin{lemma}
For any operators $A_a^{(i)}$ and any vector $|v\rangle$, one has:
\begin{equation}
	\sum_i ||\sum_a A_a^{(i)}|v\rangle||^2 \le \left[\sum_a \sqrt{\sum_i ||A_a^{(i)} |v \rangle||^2} \right]^2 \nonumber
\end{equation}
\label{lemma1}
\end{lemma}

\begin{proof}
	The proof is elementary. We first apply the triangle inequality to obtain: $||\sum_a A_a^{(i)}|v\rangle||^2 \le \sum_{a,b} [||A_a^{(i)} |v\rangle||\times ||A_b^{(i)} |v\rangle||]$. Using then Cauchy-Schwarz inequality, we have that $\sum_i [||A_a^{(i)} |v\rangle||\times ||A_b^{(i)} |v\rangle||] \le \sqrt{\sum_i ||A_a^{(i)} |v \rangle||^2} \sqrt{\sum_i ||A_b^{(i)} |v \rangle||^2}$, from which we obtain the inequality of Lemma~\ref{lemma1}.
\end{proof}

\begin{lemma}
For any operators $A_i$ and any vector $|v\rangle$, one has:
\begin{equation}
	||\sum_{i=1}^K A_i |v \rangle ||^2 \le K \sum_{i=1}^K ||A_i |v \rangle ||^2 \nonumber
\end{equation}
\label{lemma2}
\end{lemma}

\begin{proof}
	The proof is elementary. We first use the triangle inequality: $||\sum_{i=1}^K A_i |v \rangle ||^2 \le [\sum_{i=1}^K ||A_i |v \rangle ||]^2 = K^2 [\sum_{i=1}^K (1/K)||A_i |v \rangle ||]^2$. We then use the concavity of the square root function: for all $x_i$, $\sum_{i=1}^K \sqrt{x_i}/K \le \sqrt{\sum_{i=1}^K x_i /K}$. Applying this inequality for $x_i = ||A_i |v \rangle||$, we conclude the proof of Lemma \ref{lemma2}.
\end{proof}

\subsubsection{First inequality}
We have: $S_z^{(\Phi)} = U_\Phi^\dagger [\sum_i Z^{(i)} \otimes \mathbb{1}] U_\Phi = \sum_i \Phi_i^\dagger [Z^{(i)} \otimes \mathbb{1}] \Phi_i$. 
Hence, we have $S_z^{(\Phi)} - \bar{S}_z = \sum_{i=1}^N \{\Phi_i^\dagger [Z^{(i)} \otimes \mathbb{1}] \Phi_i - \mathbb{1}\otimes \bar{Z}^{(i)}\}$ and similarly $S_x^{(\Phi)} - \bar{S}_x = \sum_{i=1}^N \{\Phi_i^\dagger [X^{(i)} \otimes \mathbb{1}] \Phi_i - \mathbb{1}\otimes\bar{X}^{(i)}\}$.

Elementary manipulations of the partial SWAP gate show that:
\begin{equation}
	\Phi_i^\dagger [Z^{(i)} \otimes \mathbb{1}] \Phi_i |+\rangle_i |\Psi \rangle = |+\rangle_i \tilde{Z}^{(i)}|\Psi \rangle 
\end{equation}
And similarly that:
\begin{equation}
	\Phi_i^\dagger [X^{(i)} \otimes \mathbb{1}] \Phi_i |+\rangle_i |\Psi \rangle = \frac{1}{\sqrt{2}}[ |0\rangle_i \tilde{X}^{(i)} - |1\rangle_i \tilde{Z}^{(i)} \tilde{X}^{(i)} \tilde{Z}^{(i)}]|\Psi \rangle
\end{equation}
Applying then Lemma \ref{lemma2} with $A_i = \Phi_i^\dagger [Z^{(i)} \otimes \mathbb{1}] \Phi_i - \mathbb{1}\otimes \bar{Z}^{(i)}$ and $|v\rangle = |+\rangle^{\otimes N} |\Psi\rangle$, we have $||A_i|v\rangle||^2 = ||(\tilde{Z} - \bar{Z})^{(i)}|\Psi \rangle||^2$, and hence:
\begin{equation}
	\langle [S_z^{(\Phi)} - \bar{S}_z]^2\rangle \le N \sum_{i=1}^N ||(\tilde{Z} - \bar{Z})^{(i)}|\Psi \rangle||^2
\end{equation}
Similarly, applying Lemme \ref{lemma2} with $A_i = \Phi_i^\dagger [X^{(i)} \otimes \mathbb{1}] \Phi_i - \mathbb{1}\otimes \bar{X}^{(i)}$ and $|v\rangle = |+\rangle^{\otimes N} |\Psi\rangle$, we have $||A_i|v\rangle||^2 = (1/2)|||0\rangle_i(\bar{X} - \tilde{X})^{(i)}|\Psi \rangle + |1\rangle_i(\bar{X} + \tilde{Z} \tilde{X} \tilde{Z})^{(i)}|\Psi \rangle||^2 = (1/2)||(\bar{X} - \tilde{X})^{(i)}|\Psi \rangle||^2 + (1/2)||(\bar{X} + \tilde{Z} \tilde{X} \tilde{Z})^{(i)}|\Psi \rangle||^2$. Hence, by Lemma \ref{lemma2}:
\begin{eqnarray}
	\langle [S_x^{(\Phi)} - \bar{S}_x]^2\rangle \le \frac{N}{2} \sum_{i=1}^N \left[||(\bar{X} - \tilde{X})^{(i)}|\Psi \rangle||^2 \right.\nonumber \\ 
	\left.+ ||(\bar{X} + \tilde{Z} \tilde{X} \tilde{Z})^{(i)}|\Psi \rangle||^2\right]
\end{eqnarray}

\subsubsection{Removing the regularized operators}
As a second step, we derive inequalities which do not involve the regularized (unitary) operators $\tilde{X}^{(i)}$ and $\tilde{X}^{(i)}$. Specifically, since $\tilde{Z}^{(i)}$ is unitary and $[\tilde{Z}^{(i)}]^2=\mathbb{1}$, we have $||(\tilde{Z} - \bar{Z})^{(i)}|\Psi \rangle|| = ||(\mathbb{1} - \tilde{Z}\bar{Z})^{(i)}|\Psi \rangle||$. Then, we use that $\mathbb{1} - \tilde{Z}\bar{Z} = \mathbb{1} - |\bar{Z}| \le (\mathbb{1} - |\bar{Z}|)(\mathbb{1} + |\bar{Z}|) = \mathbb{1} - \bar{Z}^2$. Therefore: $||(\tilde{Z} - \bar{Z})^{(i)}|\Psi \rangle||^2 \le ||(\mathbb{1} - \bar{Z}^2)^{(i)}|\Psi \rangle||^2$. Similarly, we have that $||(\tilde{X} - \bar{X})^{(i)}|\Psi \rangle||^2 \le ||(\mathbb{1} - \bar{X}^2)^{(i)}|\Psi \rangle||^2$. We also have that $||(\bar{X} + \tilde{Z} \tilde{X} \tilde{Z})^{(i)}|\Psi \rangle|| = ||(\bar{X}\tilde{Z} + \tilde{Z} \tilde{X} )^{(i)}|\Psi \rangle||$. Inserting $\tilde{Z}=\bar{Z} + \tilde{Z} - \bar{Z}$ and $\tilde{X}=\bar{X} + \tilde{X} - \bar{X}$, we obtain $\bar{X}\tilde{Z} + \tilde{Z} \tilde{X} = \bar{X}\bar{Z} + \bar{Z} \bar{X} + (\tilde{Z} - \bar{Z}) \bar{X} + (\tilde{X} - \bar{X}) \bar{Z} + \tilde{X}(\tilde{Z} - \bar{Z})$. Applying then Lemma \ref{lemma1}, we have:
\begin{eqnarray}
	&\sum_{i=1}^N||(\bar{X} + \tilde{Z} \tilde{X} \tilde{Z})^{(i)}|\Psi \rangle||^2 \le& \left[
		\sqrt{\sum_i ||(\bar{Z}\bar{X} + \bar{X}\bar{Z})^{(i)}|\Psi \rangle||^2} \right. \nonumber \\
		&+ \sqrt{\sum_i ||[(\tilde{Z} - \bar{Z})\bar{X}]^{(i)}|\Psi \rangle||^2}  
		&+ \sqrt{\sum_i ||[(\tilde{X} - \bar{X})\bar{Z}]^{(i)}|\Psi \rangle||^2}  \nonumber \\
		&+\left. \sqrt{\sum_i ||[\tilde{X}(\tilde{Z} - \bar{Z})]^{(i)}|\Psi \rangle||^2} \right]^2&
\end{eqnarray}
We then use the fact that $\tilde{X}^{(i)}$ is unitary (implying that $||\tilde{X}^{(i)}||=1$), and bound the norm of operators $\bar{X}^{(i)}$ and $\bar{Z}^{(i)}$ in the following way:
 $||\bar{X}^{(i)}|| = \frac{2}{k} ||\sum_{a=0}^{k-1} \sin(a\pi/k) \sigma_a^{(i)} ||
  \le  \frac{2}{k} \sum_{a=0}^{k-1} \sin(a\pi/k)  =  \frac{\sin(\pi/k)}{k\sin^2[\pi / (2k)]}  \le 4 / \pi$. By a similar computation, we have $||\bar{Z}^{(i)}|| = \frac{2}{k} ||\sum_{a=0}^{k-1} \cos(a\pi/k) \sigma_a^{(i)} ||
  \le  \frac{2}{k} \sum_{a=0}^{k-1} |\cos(a\pi/k)|  =\frac{2\sin[(\pi/k)\lfloor k/2 \rfloor - 1/2]}{k\sin[\pi / (2k)]} \le 4 / \pi$. Using these bounds, we obtain:
\begin{eqnarray}
	\sum_{i=1}^N||(\bar{X} + \tilde{Z} \tilde{X} \tilde{Z})^{(i)}|\Psi \rangle||^2 &\le& \left[
		\sqrt{\sum_i ||(\bar{Z}\bar{X} + \bar{X}\bar{Z})^{(i)}|\Psi \rangle||^2} \right. \nonumber \\
		&+ &\left(1 + \frac{4}{\pi}\right) \sqrt{\sum_i ||(\tilde{Z} - \bar{Z})^{(i)}|\Psi \rangle||^2}  \nonumber \\
		  &+ & \left. \frac{4}{\pi}\sqrt{\sum_i ||(\tilde{X} - \bar{X})^{(i)}|\Psi \rangle||^2}  \right]^2
\end{eqnarray}

Putting everything together, we finally obtain:
\begin{equation}
	\langle [S_z^{(\Phi)} - \bar{S}_z]^2\rangle \le N \sum_{i=1}^N ||(\mathbb{1}- \bar{Z}^2)^{(i)}|\Psi \rangle||^2
\end{equation}
\begin{eqnarray}
	\langle [S_x^{(\Phi)} - \bar{S}_x]^2\rangle & \le & \frac{N}{2} \sum_{i=1}^N||(\mathbb{1} - \tilde{X}^2)^{(i)}|\Psi \rangle||^2\nonumber \\ 
	&+& \frac{N}{2}\left[
		\sqrt{\sum_{i=1}^N ||(\bar{Z}\bar{X} + \bar{X}\bar{Z})^{(i)}|\Psi \rangle||^2} \right. \nonumber \\
		&+ &\left(1 + \frac{4}{\pi}\right) \sqrt{\sum_{i=1}^N ||(\mathbb{1} - \bar{Z}^2)^{(i)}|\Psi \rangle||^2}  \nonumber \\
		  &+ & \left. \frac{4}{\pi}\sqrt{\sum_{i=1}^N ||(\mathbb{1} - \bar{X}^2)^{(i)}|\Psi \rangle||^2}  \right]^2
\end{eqnarray}

\subsubsection{Form of the robustness bound}
Our next step will be to find constants $\alpha_0$ and $\alpha_1$ such that:
\begin{eqnarray}
	\sum_{i=1}^N ||(\bar{Z}\bar{X} + \bar{X}\bar{Z})^{(i)}|\Psi \rangle||^2 \le N\epsilon \alpha_0 \\
	\sum_{i=1}^N ||(\mathbb{1} - \bar{Z}^2)^{(i)}|\Psi \rangle||^2 \le N\epsilon \alpha_1 \\
	\sum_{i=1}^N ||(\mathbb{1} - \bar{X}^2)^{(i)}|\Psi \rangle||^2 \le N\epsilon \alpha_1 
\end{eqnarray}
This will imply the following bounds:
\begin{eqnarray}
	\langle [S_z^{(\Phi)} - \bar{S}_z]^2\rangle \le N^2 \epsilon \alpha_1 \\
	\langle [S_x^{(\Phi)} - \bar{S}_x]^2\rangle \le \frac{N^2\epsilon}{2}\left\{ \nonumber \right. \\
		\left. \alpha_1 + [\sqrt{\alpha_0} + (1 + 8/\pi)\sqrt{\alpha_1}]^2 \right\} ~.
\end{eqnarray}

Finding the constants $\alpha_0$ and $\alpha_1$ will be achived using the condition $\langle {\cal B} + Nk \rangle \le Nk\epsilon$, which implies, via the SOS decomposition of Eq.~\eqref{eq_SOS_appendix}, the following bound:
\begin{equation}
	\sum_{i=1}^N \langle \Psi | ({\bm \sigma}^T [\mathbb{1} - M] {\bm \sigma} )^{(i)}|\Psi \rangle \le Nk \epsilon ~. \label{eq_ineq_SOS}
\end{equation}
We will use this inequality to derive bounds of the form: 
\begin{eqnarray}
	\sum_{j=1}^N ||(\mathbb{1} - \frac{\bar{Z}^2 + \bar{X}^2}{2})^{(j)}|\Psi \rangle||^2 \le N\epsilon 
		\label{eq_ineq_from_SOS_1}\\
	\sum_{j=1}^N ||[(\bar{Z} \pm i \bar{X})^{(j)}]^2|\Psi \rangle||^2 \le N\epsilon \alpha
	\label{eq_ineq_from_SOS_2}
\end{eqnarray}
for a certain constant $\alpha$. Introducing the notations $A_0^{(j)} = (\mathbb{1} - \frac{\bar{Z}^2 + \bar{X}^2}{2})^{(j)}$ and $A_{\pm}^{(j)} = [(\bar{Z} \pm i \bar{X})^{(j)}]^2 = [\bar{Z}^2 - \bar{X}^2 \pm i(\bar{X}\bar{Z} + \bar{Z} \bar{X})]^{(j)}$, we then use the decompositions: $[\mathbb{1} - \bar{Z}^2]^{(j)} = [A_0 - \frac{A_+ + A_-}{4}]^{(j)}$; $[\mathbb{1} - \bar{X}^2]^{(j)} = [A_0 + \frac{A_+ + A_-}{4}]^{(j)}$ and $[\bar{Z} \bar{X} + \bar{X} \bar{Z}]^{(j)} = [\frac{A_+ - A_-}{2i}]^{(j)}$. Applying then Lemma \ref{lemma1}, we obtain $\alpha_0 = \frac{1}{4}(\sqrt{\alpha} + \sqrt{\alpha})^2 = \alpha$, and similarly $\alpha_1 = (1 + \sqrt{\alpha}/2)^2$. In the next subsection, we will establish Eqs.~\eqref{eq_ineq_from_SOS_1}-\eqref{eq_ineq_from_SOS_2}. 

\subsubsection{Exploiting the SOS decomposition of the Bell's inequality}
We first use the fact that $({\bm \sigma}^T [\mathbb{1} - M] {\bm \sigma} )^{(i)} = k\mathbb{1} - (k/2)(\tilde{X}^2 + \tilde{Z}^2)^{(i)}$. The matrix $M$ being a projector, $1-M$ is also a projector, so that we have $0 \le \mathbb{1} - (1/2)(\tilde{X}^2 + \tilde{Z}^2)^{(i)} \le \mathbb{1}$. This implies that $\langle \langle \Psi | [\mathbb{1} - (1/2)(\tilde{X}^2 + \tilde{Z}^2)^{(i)}]|\Psi \rangle \le ||\mathbb{1} - (1/2)(\tilde{X}^2 + \tilde{Z}^2)^{(i)}|\Psi \rangle||^2$. Therefore, from Eq.~\eqref{eq_ineq_SOS}, we have the announced inequality Eq.~\eqref{eq_ineq_from_SOS_1}.
To obtain Eq.~\eqref{eq_ineq_from_SOS_2}, we first use the property $(1-M)^2=(1-M)$, together with the explicit expression of the matrix $M_{ab} = (2/k)\cos[(a-b)\pi/k]$, to rewrite Eq.~\eqref{eq_ineq_SOS} as:
\begin{equation}
	\sum_{j=1}^N \sum_{a=0}^{k-1} ||\{\sigma_a^{(j)} - \frac{2}{k}\sum_{b=0}^{k-1}\cos\left(\pi \frac{a-b}{k}\right)\sigma_b^{(j)}\} |\Psi \rangle ||^2 \le Nk\epsilon ~.
\end{equation}
We then decompose: 
\begin{eqnarray}
	\frac{2}{k}\sum_{b=0}^{k-1}\cos\left(\pi \frac{a-b}{k}\right)\sigma_b^{(j)} = \nonumber \\ \frac{e^{ia\pi/k}}{2}(\bar{Z} - i\bar{X})^{(j)} + \frac{e^{-ia\pi/k}}{2}(\bar{Z} + i\bar{X})^{(j)} ~.
\end{eqnarray}
We introduce the notation:
\begin{equation}
\label{eq_def_Ra}
	R_a^{(j)} :=2 \sigma_a^{(j)} e^{ia\pi/k} - (\bar{Z} + i\bar{X})^{(j)} ~,
\end{equation}
so that Eq.~\eqref{eq_ineq_SOS} can be rewritten as:
\begin{equation}
	\sum_{j=1}^N \sum_{a=0}^{k-1} ||R_a^{(j)} - e^{2ia\pi/k}(\bar{Z}-i\bar{X})^{(j)}|\Psi \rangle ||^2 \le 4Nk\epsilon ~.
\end{equation}
We then introduce the operators:
\begin{equation}
	A_a^{(j)} := R_a^{(j)} R_a^{(j)} - R_a^{(j)}e^{2ia\pi/k}(\bar{Z}-i\bar{X})^{(j)} ~.
\end{equation}
We note that the operator norm of $R_a^{(j)}$ is bounded. Indeed, we have $(\bar{Z} + i\bar{X})^{(j)} = (2/k) \sum_{a=0}^{k-1} e^{ia\pi/k} \sigma_a^{(j)}$. Hence, we have $||(\bar{Z} + i\bar{X})^{(j)}|| \le 2$ (by the triangle inequality, and the property $||\sigma_a^{(j)}||=1$). Consequently, we have $||R_a^{(j)}|| \le 4$. Therefore, we have the inequality:
\begin{equation}
	\sum_{j=1}^N \sum_{a=0}^{k-1} ||A_a^{(j)}|\Psi \rangle||^2 \le 64 Nk\epsilon ~.
\end{equation}
We have the property that: 
\begin{equation}
	\sum_{a=0}^{k-1} R_a^{(j)} R_a^{(j)} = -k[\bar{Z}^{(j)} + i\bar{X}^{(j)}]^2 ~,
\end{equation}
which implies that:
\begin{equation}
	-k[\bar{Z}^{(j)} + i\bar{X}^{(j)}]^2 = \sum_{a=0}^{k-1} A_a^{(j)} + \sum_{a=0}^{k-1}R_a^{(j)}e^{2ia\pi/k}(\bar{Z}-i\bar{X})^{(j)} ~.
\end{equation}
From Lemma \ref{lemma1}, we have the following bound:
\begin{eqnarray}
	\sum_{j=1}^N ||k[\bar{Z}^{(j)} + i\bar{X}^{(j)}]^2|\Psi \rangle||^2 \le \left[ 
	\sqrt{\sum_{j=1}^N||\sum_{a=0}^{k-1} A_a^{(j)}|\Psi \rangle||^2} \right. \nonumber \\
	\left. + \sqrt{\sum_{j=1}^N||\sum_{a=0}^{k-1} R_a^{(j)}e^{2ia\pi/k}(\bar{Z}-i\bar{X})^{(j)}|\Psi \rangle||^2} \right]^2 \nonumber ~.
\end{eqnarray}
Furthermore, from Lemma \ref{lemma2}, we have the bound:
\begin{eqnarray}
	\sum_{j=1}^N||\sum_{a=0}^{k-1} A_a^{(j)}|\Psi \rangle||^2 &\le& \sum_{j=1}^Nk\sum_{a=0}^{k-1}||A_a^{(j)}|\Psi \rangle||^2 \nonumber \\
	&\le & 64 Nk^2\epsilon
\end{eqnarray}
It remains to bound $\sum_{j=1}^N||\sum_{a=0}^{k-1} R_a^{(j)}e^{2ia\pi/k}(\bar{Z}-i\bar{X})^{(j)}|\Psi \rangle||^2$. First, we decompose $R_a^{(j)} = e^{2ia\pi/k}(\bar{Z} - i\bar{X})^{(j)} + R_a^{(j)} - e^{2ia\pi/k}(\bar{Z} - i\bar{X})^{(j)}$. Therefore, we have:
\begin{eqnarray}
	\sum_{a=0}^{k-1} R_a^{(j)}e^{2ia\pi/k}(\bar{Z}-i\bar{X})^{(j)} = \sum_{a=0}^{k-1} e^{4ia\pi/k}[(\bar{Z}-i\bar{X})^{(j)}]^2 \nonumber \\
	+ \sum_{a=0}^{k-1} e^{2ia\pi/k}[R_a - e^{2ia\pi/k}(\bar{Z}-i\bar{X})]^{(j)} (\bar{Z}-i\bar{X})^{(j)} \nonumber
\end{eqnarray} 
We have $\sum_{a=0}^{k-1} e^{4ia\pi/k} = 0$ for all $k\ge 3$. Furthermore, using Lemma \ref{lemma2} and the bound $||(\bar{Z}-i\bar{X})^{(j)} || \le 2$, we conclude that:
\begin{eqnarray}
	\sum_{i=1}^N||\sum_{a=0}^{k-1} e^{2ia\pi/k} R_a^{(j)}(\bar{Z}-i\bar{X})^{(j)} |\Psi \rangle ||^2 \le \nonumber 
	\\
	\sum_{i=1}^N 4k \sum_{a=0}^{k-1} ||[R_a - e^{2ia\pi/k}(\bar{Z}-i\bar{X})]^{(j)}|\Psi \rangle ||^2  \nonumber \\
	\le  16Nk^2 \epsilon
\end{eqnarray}
Putting everything together, we have:
\begin{equation}
	\sum_{j=1}^N ||[\bar{Z}^{(j)} + i\bar{X}^{(j)}]^2|\Psi \rangle||^2 \le 144 N \epsilon ~.
\end{equation}
Exchanging the role of $i$ and $-i$ from Eq.~\eqref{eq_def_Ra} onwards, one establishes the similar inequality:
\begin{equation}
	\sum_{j=1}^N ||[\bar{Z}^{(j)} - i\bar{X}^{(j)}]^2|\Psi \rangle||^2 \le 144 N \epsilon ~.
\end{equation}
Namely, we have proved Eq.~\eqref{eq_ineq_from_SOS_2} with the constant $\alpha=144$. Consquently, we have $\alpha_1 = (1 + \sqrt{\alpha}/2)^2 = 49$. Finally, we have established the following inequalities:
\begin{eqnarray}
	\langle [S_z^{(\Phi)} - \bar{S}_z]^2\rangle \le 49 N^2 \epsilon \\
	\langle [S_x^{(\Phi)} - \bar{S}_x]^2\rangle \le \frac{N^2\epsilon}{2}[49 + (19 + \frac{56}{\pi})^2] ~.
\end{eqnarray}
This implies, via Eq.~\eqref{eq_quantity_Q}, the bound on total spin fluctuations after the SWAP gate:
\begin{equation}
	\langle [S_x^{(\Phi)}]^2 + [S_z^{(\Phi)}]^2 \rangle \le N^2 \epsilon \left[\sqrt{2/N} + \sqrt{r}\right]^2
\end{equation}
with $r=49 + [49 + (19 + \frac{56}{\pi})^2]/2 \approx 752$

\section{Generalization with arbitrary local phases}
\label{app_phases}
 In this section, we show how the derivation presented in the main text can be extended to the following  quantum Bell's inequality: 
\begin{equation}
	{\cal B} =  \frac{2}{k}\sum_{a,b=0}^{k-1} \sum_{i \neq j} \langle \sigma_a^{(i)} \sigma_b^{(j)} \rangle \cos[\pi(a-b)/k + \phi_i - \phi_j] \ge -Nk  ~,
	\label{eq_BI_phases_app}
\end{equation}
where $\phi_i$ are local phases. We show that reaching the quantum bound $-Nk$ self-tests the property: $\sum_{i,j}\cos(\phi_i - \phi_j)\langle \hat X^{(i)}\hat X^{(j)} + \hat Z^{(i)} \hat Z^{(j)}\rangle=0$, and local measurements which are the same as for $\phi_i=0$ (equally-spaced qubit measurements in the $xz$ plane). 

In order to do so, similarly to what was done in the main text for $\phi_i=0$, we introduce the collective variables $\bar{S_a} = \sum_{i=j}^N \hat \sigma_a^{(j)} e^{i\phi_j}$. We have again: $S_{ab} = \langle \bar{S_a} \bar{S_b} \rangle - \sum_i \langle \hat \sigma_a^{(i)} \hat \sigma_b^{(i)}\rangle$. We introduce the vector notations $[{\bm \sigma}^{(i)}]^T=(\hat\sigma_0, \cdots \hat\sigma_{k-1})^{(i)}$ and $\bar{\bf S}^\dagger = (\bar{S}_0^\dagger, \ldots \bar{S}_{k-1}^\dagger)$, and the matrix $M_{ab} = (2/k)\cos[\pi(a-b)/k]$. Eq.~\eqref{eq_BI_phases_app} then takes the same expression as in the main text for $\phi_i=0$:
\begin{equation}
	{\cal B}= \langle \bar{\bf S}^\dagger M \bar{\bf S} \rangle - \sum_{i=1}^N \langle {\bm \sigma}^T M {\bm \sigma}\rangle^{(i)} ~.
\end{equation}
The matrix $M$ is diagonalized as in the main text; and again, we have $\langle {\cal B}\rangle \ge (-2N)/[k\sin^2(\frac{\pi}{2k})]$ for all classical local-variable models.

The classical bound can be violated as follows. Measuring a $N$-qubit quantum state along directions $\hat\sigma_a^{(i)} = \hat Z^{(i)} \cos \theta_a + \hat X^{(i)} \sin \theta_a$ yields $ {\cal B} = 4 \sum_{ab} M_{ab} \langle [\hat J'_a]^\dagger \hat J'_b \rangle - N \sum_{ab} M_{ab} \cos(\theta_a - \theta_b)$, where $\hat J_a' = (1/2)\sum_j \hat \sigma_a^{(j)}e^{i\phi_j}$. Choosing $\theta_a = a\pi /k$, we obtain ${\cal B}  = 2k{\cal S} - (Nk/2) \sum_{ab}M_{ab}^2$, where ${\cal S} = \sum_{i,j}\cos(\phi_i - \phi_j)\langle \hat X^{(i)} \hat X^{(j)} + \hat Z^{(i)} \hat Z^{(j)}\rangle$. Whenever ${\cal S}=0$, we obtain ${\cal B} = -(Nk/2){\rm Tr}(M^2) = -Nk$. We will show that $B_{\rm q} = -Nk$ is the maximal violation allowed by quantum physics, and that reaching this value self-tests ${\cal S}=0$.

We parallel the derivation presented in the main text. We have the sum-of-squares (SOS) decomposition (recall that $\mathbb{1}-M = (\mathbb{1}-M)^2$):
\begin{equation}
	\hat{\cal B} + Nk\mathbb{1} = ({\bf c}^T \bar{\bf S})^\dagger ({\bf c}^T \bar{\bf S}) + ({\bf s}^T \bar{\bf S})^\dagger ({\bf s}^T \bar{\bf S})
	+ \sum_{i=1}^N [{\bm \sigma}^T (\mathbb{1} - M) {\bm \sigma}]^{(i)}
\end{equation}
We can introduce the operators on which the self-testing procedure rely. The $\bar{Z}^{(i)}$ and $\bar{X}^{(i)}$ have the same expression as in the main text:
\begin{subequations}
\label{eq_def_barZX_app_phases}
\begin{align}
&	\bar{Z}^{(i)} := \sqrt{2/k} ~ {\bf c}^T {\bm \sigma}^{(i)} = \frac{2}{k}\sum_{a=0}^{k-1} \hat  \sigma_a^{(i)} \cos\left(\frac{a\pi}{k}\right) \label{eq_def_barZ_app_phases}\\
&	\bar{X}^{(i)} := \sqrt{2/k} ~ {\bf s}^T {\bm \sigma}^{(i)} = \frac{2}{k}\sum_{a=0}^{k-1}\hat  \sigma_a^{(i)}  \sin\left(\frac{a\pi}{k}\right)\label{eq_def_barX_app_phases} ~.
\end{align}
\end{subequations}
Introducing then the (hermitian) operators:
\begin{subequations}
\label{eq_def_Sxz_app_phases}
\begin{align}
&\bar S_z = \sum_{i=1}^N [ \cos \phi_i\bar{Z}^{(i)} - \sin \phi_i \bar{X}^{(i)}]\\
&\bar S_x = \sum_{i=1}^N [ \cos \phi_i\bar{X}^{(i)} + \sin \phi_i \bar{Z}^{(i)}] \\
&\hat A_a^{(i)} := \hat \sigma_a^{(i)} - \frac{2}{k}\sum_{b=0}^{k-1} \hat \sigma_b^{(i)} \cos\left(\pi\frac{a-b}{k}\right) ~, \label{eq_def_Abar_app_phases}
\end{align}
\end{subequations}
the SOS decomposition reads explicitly:
\begin{equation}
	\hat{\cal B} + Nk\mathbb{1} = \frac{k}{2}(\bar{S}_z^2 + \bar{S}_x^2) 
	+ \sum_{i=1}^N \sum_{a=0}^{k-1} [\hat{A}_a{(i)}]^2  ~.
	\label{eq_SOS_app_phases}
\end{equation}
Notice that $\bar{S}_z^2 + \bar{S}_x^2 = \sum_{i,j} \cos(\phi_i - \phi_j)(\bar{Z}^{(i)} \bar{Z}^{(j)} + \bar{X}^{(i)} \bar{X}^{(j)})$.  This shows that measuring at equal angles in a given plane a $N$-qubit state s.t. $\sum_{i,j}\cos(\phi_i - \phi_j)\langle \hat X^{(i)} \hat X^{(j)} + \hat Z^{(i)} \hat Z^{(j)}\rangle=0$ reaches the quantum bound $-Nk$. To prove the converse statement, we may follow exactly the same proof as for $\phi_i=0$. To establish the fact that $\bar{Z}^{(i)}$ and $\bar{X}^{(i)}$ act as the Pauli $\hat Z$ and $\hat X$ operator, we follow the derivation of the main text. The only difference is between Eqs.~(9) and (10) of the main text, where $\Sigma_Z^{(j)}$ and $\Sigma_X^{(j)}$ [defined before Eq.~(10)] are now: $\Sigma_Z^{(j)} = \sum_{j'\neq j} [ \cos \phi_{j'}\bar{Z}^{(j')} - \sin \phi_{j'} \bar{X}^{(j')}]$ and $\Sigma_X^{(j)} = \sum_{j'\neq j} [ \cos \phi_{j'}\bar{X}^{(j')} + \sin \phi_{j'} \bar{Z}^{(j')}]$. From the SOS being zero, we have that $(\Sigma_Z -i\Sigma_X)^{(j)}|\psi\rangle = -e^{-i\phi_j}(\bar{Z} - i\bar{X})^{(j)}|\psi \rangle$. As the operator $\Sigma_Z^{(j)} - i\Sigma_X^{(j)}$ acts on subsystems $j'\neq j$, it commutes with any operator acting on $j$. The rest of the proof goes as in the main text. We finally construct the SWAP gate on the $\bar{Z}^{(i)}$ and $\bar{X}^{(i)}$ operators. To conclude, we notice that maximal violation of the BI implies that $\bar{S}_z \bar{\rho} = 0 = \bar{S}_x \bar{\rho}$. Hence the same property after the partial-SWAP gate: $\hat J'_x \hat \rho  = \hat J'_z \hat \rho = 0$, where $\hat J'_x = (1/2)\sum_{i=1}^N [\cos\phi_i\hat X^{(i)} + \sin\phi_i \hat Z^{(i)}]$ and $\hat J_z' = (1/2)\sum_{i=1}^N [\cos\phi_i\hat Z^{(i)} - \sin\phi_i \hat X^{(i)}]$. Hence, maximal quantum violation of the BI self-tests $4\langle (J_x')^2 + (J_z')^2 \rangle = 0 = \sum_{i,j}\cos(\phi_i - \phi_j)\langle \hat X^{(i)}\hat X^{(j)} + \hat Z^{(i)} \hat Z^{(j)}\rangle$.

\end{document}